\documentclass[a4paper,UKenglish,cleveref, autoref, thm-restate,svgnames]{lipics-v2021}

\hideLIPIcs  

\title{A Little Aggression Goes a Long Way} 

\author{Jyothi Krishnan}{IIT Gandhinagar, India}{jyothi.k@iitgn.ac.in }{}{Supported by the Center for Creative Learning}

\author{Neeldhara Misra}{IIT Gandhinagar, India}{neeldhara.m@iitgn.ac.in }{}{Supported by the SERB Early Career Researcher Grant ECR/2018/002967}

\author{Saraswati Girish Nanoti}{IIT Gandhinagar, India}{nanoti_saraswati@iitgn.ac.in}{}{Supported by CSIR}

\authorrunning{J. Krishnan, N. Misra, and S. Nanoti} 

\Copyright{} 
    
\ccsdesc[500]{Theory of computation~Design and analysis of algorithms}

\keywords{games, complexity} 

\category{} 

\relatedversion{} 

\acknowledgements{}

\nolinenumbers 

\EventEditors{John Q. Open and Joan R. Access}
\EventNoEds{2}
\EventLongTitle{42nd Conference on Very Important Topics (CVIT 2016)}
\EventShortTitle{CVIT 2016}
\EventAcronym{CVIT}
\EventYear{2016}
\EventDate{December 24--27, 2016}
\EventLocation{Little Whinging, United Kingdom}
\EventLogo{}
\SeriesVolume{42}
\ArticleNo{23}

\usepackage{charter,eulervm,parskip,tikz,float,nicefrac,framed}
\usetikzlibrary{shapes.geometric, calc, positioning, decorations.pathreplacing}

\begin{document}

\maketitle

\begin{abstract}
Aggression is a two-player game of troop placement and attack played on a map (modeled as a graph). Players take turns deploying troops on a territory (a vertex on the graph) until they run out. Once all troops are placed, players take turns attacking enemy territories. A territory can be attacked if it has~$k$ troops and there are more than~$k$ enemy troops on adjacent territories. At the end of the game, the player who controls the most territories wins. In the case of a tie, the player with more surviving troops wins. The first player to exhaust their troops in the placement phase leads the attack phase.

We study the complexity of the game when the graph along with an assignment of troops and the sequence of attacks planned by the second player. Even in this restrained setting, we show that the problem of determining an optimal sequence of first player moves is NP-complete. We then analyze the game for when the input graph is a matching or a cycle.
\end{abstract}

\section{Introduction}
\label{sec:introduction}

The game of \textsc{Aggression} was originally published in 1973 in Games with Pencil and Paper~\cite{solomongames} by mathematician and game designer Eric Solomon. The game has subsequently featured in the Leapfrogs Action Book: Doodles~\cite{1976leapfrog} written by the mathematics education group Leapfrogs, published in 1976. A variation of it --- called ``A Little Bit of Aggression'' also shows up in MathPickle's collection of activities~\cite{mathpickle}. In this game, a map is either given (see~\cref{fig:a_little_bit_of_aggression}) or drawn by the players (see~\Cref{fig:aggression}). Two players have a certain number of troops each. The game has two phases.

\begin{figure}[t]
    \centering
    \includegraphics[width=0.5\textwidth]{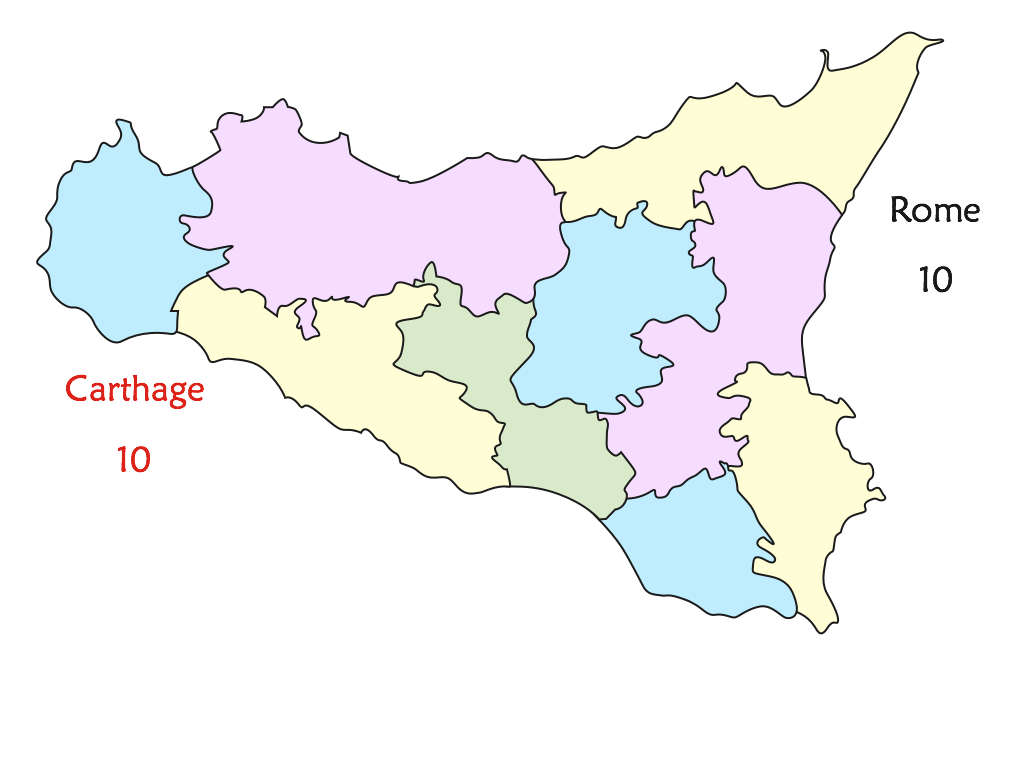}
    \caption{An example map that goes with ``A Little Bit of Aggression''. This is a map of Sicily which is played between Rome and Carthage. The map is from the MathPickle website.}
    \label{fig:a_little_bit_of_aggression}
\end{figure}

In the \emph{placement phase}, players take turns choosing an empty territory and placing any number of their troops into that territory. troops do not move once assigned to a territory. If a player has no troops left or if there are no empty territories on the board - they pass. The placement phase continues until both players pass. In the \emph{attacking phase}, the player who passed first in the placement phase has the first move. Players alternate selecting an enemy territory and counting all of their neighboring troops. If their combined strength is greater than the number of troops in the enemy territory, the enemy troops are all destroyed. This phase continues until no further attacks are possible.

The player who controls the most territories wins. In the case of a tie, we count the troops --- the player with more of them at the end of the game wins. If the number of troops are equal as well, the game is declared a draw.

\paragraph*{Our Contributions} 

The game of \textsc{Aggression} is a partisan combinatorial game, since it is a turn-based game between two players with perfect information and no element of chance, however, different players have access to different moves (each player commands their own troops). To the best of our knowledge, the game has not been studied from the perspectives of either computational complexity or combinatorial game theory.

Maps can be thought of as graphs, where the territories are interpreted as vertices, and the edges correspond to pairs of territories that share a border. Indeed, maps correspond to planar graphs --- however, we study \textsc{Aggression} on graphs in general, since the rules carry over naturally. 

Some graph structures admit straightforward strategies for the players:

\begin{itemize}
    \item If the graph is complete, i.e, all possible edges are present, it is optimal for both players to position all their troops on a single vertex, and this game ends in a draw.
    \item If the graph is a star, i.e, we have one ``central'' node~$v$ and all other nodes are adjacent to~$v$ and nothing else, it is again optimal for both players to position all their troops on a single vertex, and this game ends in a draw.
\end{itemize}

However, even for matchings, paths, and cycles, analyzing player outcomes turns out to be non-trivial. We show the following.

\begin{itemize}
    \item Both players can draw the game on matchings of size at most two.
    \item Both players can draw the game played on a matching if the number of edges in the matching is at least the number troops. 
    \item The second player has a winning strategy for the game played on a matching of size at least three when the number of troops is at least six more than the number of edges. 
    \item Both players have a strategy to draw the game on cycles of length three, four, and five.
\end{itemize}

We also formulate the following algorithmic question inspired by the setting of the game. Given a graph and an assignment of troops of both players to the vertices of the graph, and a sequence of attacks planned by the second player on their turn, can the first player win the game? We show this restricted non-adaptive version of the second phase of the original game is already NP-complete, even on bipartite graphs.

\textbf{Related Work.} \textsc{Aggression} is an example of a game on a graph. Games on graphs are often categorized as pursuit-evasion such as cops and robbers games~\cite{bonato2011game}, strategic defense games like greedy spiders~\cite{DBLP:conf/fun/HaanW18}, and guarding games~\cite{FOMIN20116484}. We refer the reader to~\cite{fijalkow2023games} for a detailed survey on games on graphs. We are not aware of prior literature specific to the game of aggression.

\begin{figure}[t]
    \centering
    \setlength{\fboxsep}{10pt} 
    \fbox{
    \begin{minipage}{0.45\textwidth}  
        \centering
        \includegraphics[width=\textwidth]{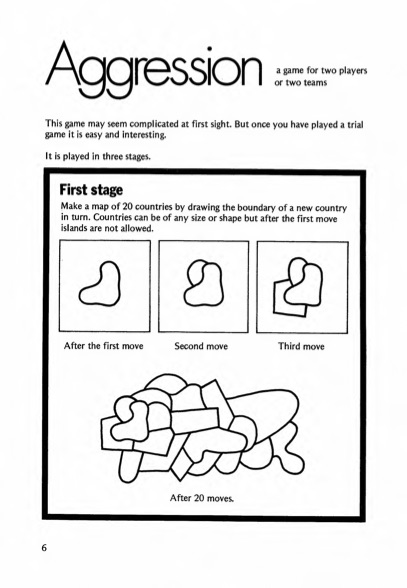}  
    \end{minipage}\hfill
    \begin{minipage}{0.45\textwidth}  
        \centering
        \includegraphics[width=\textwidth]{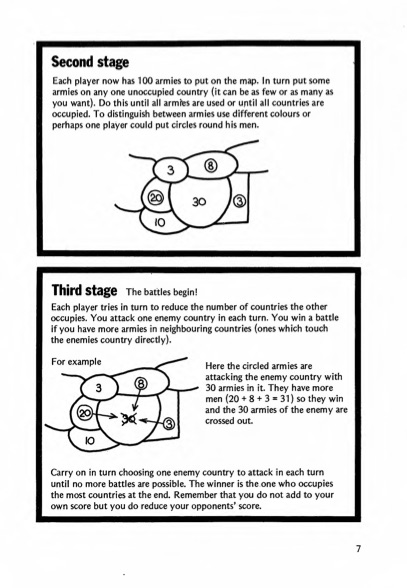}  
    \end{minipage}
    }
    \caption{Aggresion as depicted in the Leapfrogs Action Book ``Doodles''.}
    \label{fig:aggression}
\end{figure}

\section{Preliminaries and Notation}
\label{sec:prelims}


In the game of aggression, we have a graph~$G = (V,E)$, budgets~$T_L$ and~$T_R$, and two players who we call Lata (the first player) and Raj (the second player). The game proceeds in two phases:

\textbf{Placement Phase.} Lata and Raj take turns placing their troops on the vertices of the graph. Lata goes first. We use~$f_1(v)$ to denote the number of troops placed by Lata on vertex~$v$ and~$f_2(v)$ to denote the number of troops placed by Raj on vertex~$v$. Note that:

\[ \sum_{v \in V} f_1(v) \leqslant T_L; \sum_{v \in V} f_2(v) \leqslant T_R, \]

i.e, Lata can place a total of at most~$T_L$ troops and Raj can place a total of at most~$T_R$ troops; and:

\[ f_1(v) > 0 \implies f_2(v) = 0; f_2(v) > 0 \implies f_1(v) = 0, \]

i.e, one vertex can host troops from at most one of the players. A pair of functions that satisfy the properties above is called a two-player~$(T_L,T_R)$-troop placement. If~$T_L = T_R = T$, as will typically be the case, we call this a two-player~$T$-troop placement. We call a vertex~$v$ a player 1 vertex or Lata's vertex if~$f_1(v) > 0$, and similarly, a player 2 vertex or Raj's vertex if~$f_2(v) > 0$. We call a vertex~$v$ a neutral vertex if~$f_1(v) = f_2(v) = 0$.

\textbf{Attack Phase.} The first player to exhaust their troops in the placement phase gets the first move in the attack phase. Given~$f_1, f_2$, a vertex is vulnerable if either:

\[ f_1(v) > 0 \mbox{ and } \sum_{u \in N(v)} f_2(u) > f_1(v) \]

or

\[ f_2(v) > 0 \mbox{ and } \sum_{u \in N(v)} f_1(u) > f_2(v), \]

where~$N(v)$ denotes the set of neighbors of~$v$. The first player to move can choose a vulnerable vertex that ``belongs'' to the other player and attack it. If the attack is successful, the troops at the vertex are eliminated. The second player to move can then choose a vulnerable vertex that belongs to the opponent and attack it. Note that after a vertex is attacked, it is emptied of the troops that it has and becomes a neutral territory. In particular, the attacking player does not lose any troops in the process, nor does it ``gain'' the attacked territory. An attacking player simply aims to reduce the number of territories owned by the opponent.

This phase continues until no attacks are possible. At the end, the player with more territories wins, and if there is a tie, the player with more troops on the board wins --- in other words, in the event of a tie on the territories, player~$b \in \{1,2\}$ wins if~$\sum_{v \in V} f_b(v) > \sum_{v \in V} f_{3-b}(v)$. If both players have the same number of territories and troops at the end of the attack phase, the game is a draw. A player is said to have a \textbf{strong win} if they have at least two territories more than their opponent at the end of the game.

\paragraph*{An Algorithmic Question Inspired by Aggression} In \textsc{Optimal Response}, we are given a two-player~$T$-troop placement and an attack plan for Raj, and we want to know if Lata has a winning strategy against this particular attack plan:

\begin{framed}
\textsc{Optimal Response}
\begin{itemize}
    \item \textbf{Input.} A graph~$G = (V,E)$, two functions~$f_1, f_2: V \to \mathbb{N}$ such that~$(f_1,f_2)$ is a two-player~$(T_L,T_R)$-troop placement, and a sequence~$\sigma := \{i_1, i_2, \ldots, i_p\}$.
    \item \textbf{Output.} Consider the attack phase of the game where troops are placed according to~$(f_1,f_2)$. Output \textsc{Yes} if there is a sequence~$\tau := \{j_1, j_2, \ldots, j_q\}$ such that if Lata attacks~$j_k$ in her~$k\textsuperscript{th}$ move and Raj attacks~$i_k$ in his~$k\textsuperscript{th}$ move, then Lata wins the game; \textsc{No} otherwise.
\end{itemize}
\end{framed}




The \textsc{Multi-Colored Clique} problem is the following:

\begin{framed}
    \textsc{Multi-Colored Clique}
    \begin{itemize}
        \item \textbf{Input.} A graph~$G = (V_1 \uplus \cdots \uplus V_k,E)$, where~$|V_1| = \cdots = |V_k| = n$.
        \item \textbf{Output.} Output \textsc{Yes} if there is a subset~$S \subseteq V(G)$ such that~$|S \cap V_i| = 1$ for all~$1 \leqslant i \leqslant k$ and~$G[S]$ is a clique, and \textsc{No} otherwise.
    \end{itemize}
\end{framed}

\textsc{Multi-Colored Clique} is known to be \textsf{NP}-complete~\cite{pcbook}.

\section{Optimal Response}
\label{sec:or}

In this section, we show that \textsc{Optimal Response} is \textsf{NP}-complete. Membership in \textsf{NP} is immediate, since we can guess the sequence of attacks by Lata and verify that it is winning when combined with Raj's (given) attack sequence. We now demonstrate hardness by a reduction from \textsc{Multi-Colored Clique}.

\begin{theorem}
    \label{thm:or}
    \textsc{Optimal Response} is \textsf{NP}-complete even on bipartite graphs.
\end{theorem}

\begin{proof}
We reduce from \textsc{Multi-Colored Clique}. Let~$G = (V_1 \uplus \cdots \uplus V_k,E)$ be an instance of \textsc{Multi-Colored Clique}, where~$n := |V_1| = |V_2|= \cdots =|V_k|$ and~$E := \{e_1,\ldots,e_m\}$. We assume, without loss of generality, that $n > k+2$: this can always be ensured, for example, by adding isolated vertices in the color classes. We construct an instance~$$\langle H = (V,F), T, (f_1,f_2), \sigma := \{i_1, i_2, \ldots, i_p\} \rangle$$ of \textsc{Optimal Response} as follows (see also:~\Cref{fig:construction}). 

\begin{itemize}
    \item The vertex set of~$H$ has one vertex for every vertex in~$G$. We use~$U_1 \uplus \cdots \uplus U_k$ to denote these vertices, where~$u_{i,j}$ for~$1 \leqslant i \leqslant k$ and~$1 \leqslant j \leqslant n$ denotes the vertex in~$H$ corresponding to the~$j\textsuperscript{th}$ vertex in~$V_i$. It also has one vertex for every edge in~$G$. We use~$w_e$ to denote the vertex corresponding to an edge~$e \in E$. Finally, we introduce~$k$ ``guard'' vertices, denoted~$g_1, \ldots, g_k$ and a set of~$(n-1)\cdot k - {k \choose 2} + 1$ vertices denoted by~$Z$.
    \item For each~$1 \leqslant i \leqslant k$, the guard vertex~$g_i$ is adjacent to all vertices in~$U_i$.
    \item For all~$e \in E$ such that~$e = (x,y)$, where~$x$ is the~$p\textsuperscript{th}$ vertex in~$V_i$ and~$y$ is the~$q\textsuperscript{th}$ vertex in~$V_j$, we make~$w_e$ adjacent to~$u_{i,p}$ and~$v_{j,q}$.
    \item The vertices in~$Z$ are isolated.
    \item Lata has $(m+1)$ troops on each of the guard vertices, that is,~$f_1(g_i) = m+1$ for all~$1 \leqslant i \leqslant k$; one troop on all the vertices corresponding to edges of~$G$, that is,~$f_1(w_e) = 1$ for all~$e \in E$; and one troop on every vertex of~$Z$, that is,~$f_1(z) = 1$ for all~$z \in Z$.
    \item Raj has $m$ troops on each vertex in~$H$ that correspond to the vertices of~$G$, that is,~$f_2(u_{i,j}) = m$ for all~$1 \leqslant i \leqslant k$ and~$1 \leqslant j \leqslant n$,
    \item The sequence~$\sigma$ is given by:~$\{g_1, \ldots, g_k, w_{e_1}, \ldots, w_{e_m}\}$.
    \item Note that~$T_L = k(m+1) + m + (n-1)k-{k \choose 2} + 1$ and~$T_R = mnk$.
\end{itemize}

\begin{figure}[t]
    \centering
\begin{tikzpicture}
    \tikzset{
      rectnode/.style={draw, rectangle, minimum height=1cm, minimum width=4cm},
      circlenode/.style={draw, circle, white, fill=IndianRed, scale=0.75},
      starnode/.style={draw, white, star, star points=5, minimum size=1cm, fill=SeaGreen},
      Znode/.style={draw, circle, white, fill=SeaGreen, scale=1},
      squarenode/.style={draw, rectangle, minimum size=0.42cm, white, fill=SeaGreen},
    }
  
    \node[rectnode] (rect1) at (0,0) {};
    \node[rectnode] (rect2) at (5,0) {};
    \node[rectnode] (rect3) at (10,0) {};
  
    \node at (2.5,0) {$\ldots$};
    \node at (7.5,0) {$\ldots$};

    \node[starnode, above=of rect1] (star1) { $g_1$ };
    \node[starnode, above=of rect2] (star2) { $g_i$ };
    \node[starnode, above=of rect3] (star3) { $g_k$ };
  
    \node[squarenode] (square) at (3.5,-2) {};
    \draw (square) -- (5.5,0);
    \draw (square) -- (0,0);

    \foreach \i in {2,2.5,3,3.5,4,4.5,5} {
        \draw (star1) -- (\i-3.5,0);
        \node[circlenode] at (\i-3.5,0) {};
    }

    \foreach \i in {6,6.5,7,7.5,8,8.5,9} {
        \draw (star2) -- (\i-2.5,0);
        \node[circlenode] at (\i-2.5,0) {};
    }
    
    \foreach \i in {2,2.5,3,3.5,4,4.5,5} {
        \draw (star3) -- (\i+6.5,0);
        \node[circlenode] at (\i+6.5,0) {};
    }

    \node at (rect1.south west) [below left] {$V_1$};
    \node at (rect2.south west) [below left] {$V_i$};
    \node at (rect3.south west) [below left] {$V_k$};
  
    \foreach \i in {2,3,4,5,6,7,8,9} {
        \node[Znode] at (\i-1,-4) {};
    }
    \draw [decorate,decoration={brace,amplitude=10pt,mirror,raise=4pt},yshift=0pt] (0.8,-4.2) -- (8.2,-4.2) node [black,midway,yshift=-0.8cm] {};

    \node at (4.5,-5) {$Z$};

    \node at (square.south) [below] {$e_i$};
  \end{tikzpicture}
\caption{The graph $H$ constructed in the proof of \Cref{thm:or}. Vertices with Lata's troops are colored {\color{SeaGreen}green}, vertices with Raj's troops are colored {\color{IndianRed}red}. Note that all vertices that belong to Raj have $m$ troops, the guard vertices (depicted as stars) have $(m+1)$ troops, while the edge vertices (depicted as squares) and the vertices in $Z$ have one troop each.}
\label{fig:construction}
\end{figure}
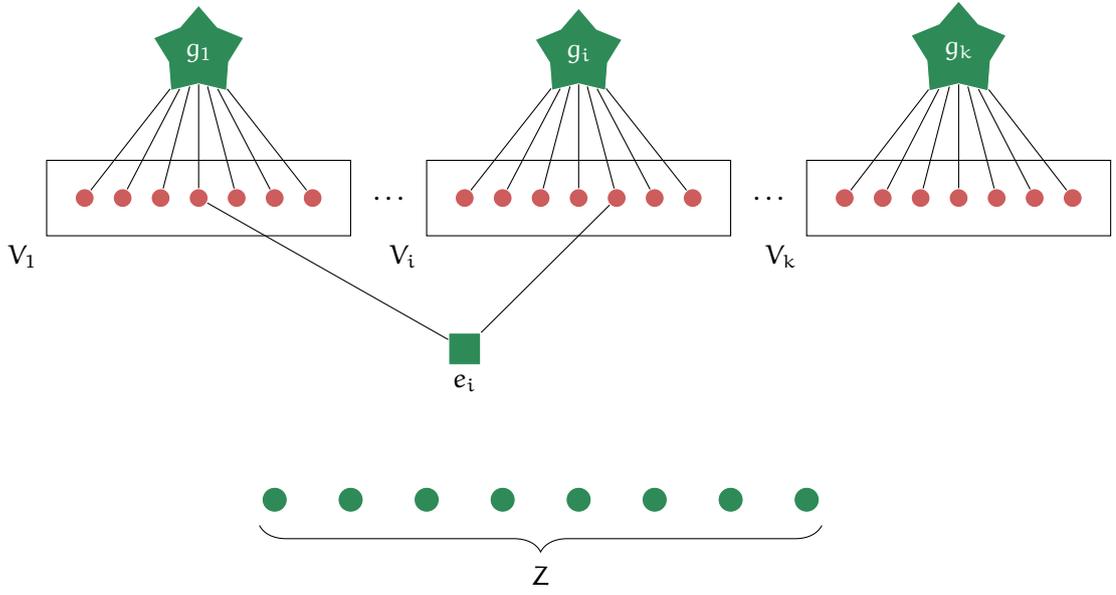

This completes a description of the construction of the reduced instance. We now argue the equivalence of the instances, i.e, we show that~$G$ has a multi-colored clique if and only if Lata has a winning response to~$\sigma$ in the game described above.

\textbf{The Forward Direction.} Suppose that~$G$ has a multi-colored clique~$S = \{x_1, \ldots, x_k\}$, where~$x_i := S \cap V_i$. Let~$f(i)$ denote the index of~$x_i$ in the part~$V_i$. Then, for~$1 \leqslant i \leqslant k$, on her~$i\textsuperscript{th}$ move, Lata attacks~$u_{i,f(i)}$ using her $(m+1)$ troops on $g_i$. Notice that after these $k$ moves, Lata has no valid attacks left: all the guard vertices are emptied of troops, and there are no threats posed by the $w_e$'s. Indeed, a $u_{i,j}$ vertex that is still occupied has $m$ troops and can be attacked by at most $m$ troops, so they remain safe. 

Now observe that all but ${k \choose 2}$ of the attacks planned by Raj according to $\sigma$ are successful. To see this, consider $e = (x,y) \in G[S]$. Let~$x$ be the~$p\textsuperscript{th}$ vertex in~$V_i$ and~$y$ be the~$q\textsuperscript{th}$ vertex in~$V_j$. Then the attack on $w_e$ is not a valid one, because there are no troops on either~$u_{i,p}$ and~$v_{j,q}$ --- recall that these were the locations attacked by $g_i$ and $g_j$, respectively.

The number of territories owned by Lata at the end of the attack phase is:

\[{\color{SeaGreen}{k \choose 2}} + {\color{OrangeRed}\left((n-1)\cdot k - {k \choose 2} + 1\right)} = (n-1)\cdot k + 1,\]

where the first term corresponds to {\color{SeaGreen}vertices in $\{w_e \mid e \in G[S]\}$} and the second term corresponds to the {\color{OrangeRed}vertices in $Z$}.

Let us now count the number of territories owned by Raj at the end of the attack phase. Raj originally started with $nk$ territories corresponding to the vertices in $U_1 \uplus \cdots \uplus U_k$. Further, Lata executed a successful attack on exactly one vertex in each $U_i$. Recalling that $|U_i| = n$ for all $1 \leqslant i \leqslant k$, we see that Raj owns $(n-1)\cdot k$ territories at the end of the attack phase.

Therefore, Lata wins this game with this response.

\textbf{The Reverse Direction.} Suppose that Lata has a winning response to $\sigma$. Observe that the first $k$ attacks in $\sigma$ will be valid no matter what attacks are made by Lata, since $n > k+2$. Also, observe that if $g_i$ has no troops, then no vertex in $U_i$ that has troops is vulnerable to an attack, since the vertices in $U_i$ have $m$ troops to begin with, and they are adjacent to at most $m$ of Lata's troops. We also know that for all $1 \leqslant i \leqslant k$, the vertex $g_i$ will not be available for attack after the $i\textsuperscript{th}$ move by Raj. Therefore, Lata has $k$ valid attacks overall.

Consider an edge $e = (x,y) \in G[S]$. Let~$x$ be the~$p\textsuperscript{th}$ vertex in~$V_i$ and~$y$ be the~$q\textsuperscript{th}$ vertex in~$V_j$. Then the attack on $w_e$ is \emph{invalid} if and only if both $u_{i,p}$ and $v_{j,q}$ have been attacked by Lata. Let us say that such an edge stands \emph{protected}. It is easy to check that the number of territories that Lata will own after the attack phase will be:

\[{\color{SeaGreen}{\ell}} + {\color{OrangeRed}\left((n-1)\cdot k - {k \choose 2} + 1\right)} = (n-1)\cdot k + 1 + \left(\ell - {k \choose 2}\right),\]

where the first term corresponds to the {\color{SeaGreen}number of protected vertices among $\{w_e \mid e \in G[S]\}$} and the second term corresponds to the {\color{OrangeRed}vertices in $Z$}.

On the other hand, since Lata has only $k$ valid attacks, Raj will own at least $nk - k = (n-1)\cdot k$ territories at the end of the attack phase. Therefore, for Lata to win, we must have that $\ell \geqslant {k \choose 2}$, in other words, the number of protected edges must be at least ${k \choose 2}$. However, every protected edge is incident to with two attacked vertices among $U_1 \uplus \cdots \uplus U_k$, and we know that there were only $k$ attacks. The only way for ${k \choose 2}$ edges to have both their endpoints lie in a set of $k$ vertices is if the $k$ vertices in question form a clique. Therefore, it must be the case that the choices of attacks made by Lata corresponds to the choice of a vertex subset that forms a multi-colored clique in $G$. 
\end{proof}

\begin{remark}
We note that it is easy to modify the reduction to ensure that $T_L = T_R$ in the reduced instance by appropriately adding dummy vertices and troops. It is also easy to ensure that the reduced instance is a connected bipartite graph, by making the isolated vertices adjacent to $g_1$, for example.
\end{remark}

\section{Winning Strategies}
\label{sec:strategies}

\subsection{Matchings}

Let~$G$ be a disjoint union of edges denoted by~$e_1 := (u_1,v_1), \ldots, e_m := (u_m,v_m)$. We first observe that the second player can always draw a game of \text{Aggression} played on~$G$ by mirroring the first player's moves on the other endpoint of the edge chosen by them.

\begin{lemma}
\label{lem:raj-draw-matching}
    If~$G$ is a disjoint union of edges, Raj has a drawing strategy.
\end{lemma}

\begin{proof}
    If Lata places~$k_i$ troops on an edge~$e$ on her~$i\textsuperscript{th}$ move, Raj responds by placing~$k_i$ troops on the unoccupied endpoint of~$e$ in response. Note that the invariant maintained at the end of~$r$ rounds is that the endpoints of a subset of~$r$ edges have been occupied by both players. As no attacks are feasible at the end of the placement phase, the game is a draw.
\end{proof}

If there are a large number of edges relative to the total number of troops, then we show that the first player also has a drawing strategy.

\begin{lemma}
    If~$G$ is a disjoint union of edges and~$m \geqslant 2T$, then Lata has a drawing strategy.
\end{lemma}

\begin{proof}
    On her~$i\textsuperscript{th}$ move in the placement phase, let~$f(i)$ denote the smallest index for which both endpoints of~$e_{f(i)}$ are unoccupied. Lata places one troop on~$e_{f(i)}$. Raj has three possible responses to this move:

    \begin{enumerate}
        \item Place one troop on an unoccupied endpoint of~$e_{f(\ell)}$ for some~$\ell \leqslant i$. With this response, we say that Raj neutralized the edge~$e_{f(\ell)}$.
        \item Place more than one troop on an unoccupied of~$e_{f(\ell)}$ for some~$\ell \leqslant i$. With this response, we say that Raj captured the edge~$e_{f(\ell)}$.
        \item Place one or more troops on an endpoint of an edge~$e$ that has not been occupied by either player. With this response, we say that Raj initiated the edge~$e$.
    \end{enumerate}

    Note that Lata never runs out of troops before Raj, since she places only one troop on each move. If Lata and Raj play the same number of moves in the placement phase, then observe that Raj has not captured any edges, and Lata and Raj have the same number of territories. There are no valid attacks in this case, so the game ends in a draw.

    Now assume that Raj made~$k$ capturing responses. This implies that:

    \begin{itemize}
        \item When Raj runs out of troops, Lata has at least~$k$ troops remaining in the placement phase.
        \item Lata and Raj occupy the same number of territories (say~$R$) when Raj runs out of troops.
        \item None of the territories occupied by Raj are vulnerable to attack.
        \item Lata occupies~$R-k$ territories that are not vulnerable in the attack phase.
    \end{itemize}

s    The number of edges for which at least one of the endpoints is occupied by troops is at most~$k + 2 \cdot (T-k)$. The number of edges that remain with both endpoints unoccupied is therefore at least:
    
$$m - \left(k+2 \cdot(T-k)\right) = m - 2T + k \geqslant k.$$

    This gives Lata~$2k$ vertices on which she can spread her remaining troops. She has at least~$k$ troops to place, and she can place them in a way that the territories they occupy will not be vulnerable in the attack phase. So the total number of territories occupied by both players at the end of the attack phase is~$R$, and the game ends in a draw.
\end{proof}

Therefore, the game ends in a draw for both players when there are a large number of edges relative to the total number of troops. However, when there are more troops than edges, then we show that the second player can, in fact, win. 

\begin{observation}
    \label{lem:2edges-draw}
    Assume Lata and Raj have~$T$ troops each. If~$G$ is a single edge or a disjoint union of two edges, then Lata has a drawing strategy.
\end{observation}

\begin{proof} The case when $G$ is a single edge is trivial. Suppose $G$ is a disjoint union of two edges, say $\{(u,v), (x,y)\}$. Lata places $\lceil \nicefrac{T}{2} \rceil$ troops on $u$. If Raj responds by playing on either $x$ or $y$, then Lata responds by placing $\lfloor \nicefrac{T}{2} \rfloor$ troops on $v$ --- in this case the game is evidently drawn since no attacks are possible and Lata claims at least two territories --- otherwise, Raj plays on $v$ and Lata responds by positioning $\lfloor \nicefrac{T}{2} \rfloor$ troops on $x$. This scenario is also a draw:
\begin{itemize}
\item If Raj plays \textbf{{\color{SeaGreen}less than $\lceil \nicefrac{T}{2} \rceil$ troops on $v$}}, then Lata can attack $v$ and Raj can attack $x$, and the game ends in a draw with both players having claim to one territory each.

\item If Raj \textbf{{\color{DodgerBlue}plays exactly $\lceil \nicefrac{T}{2} \rceil$ troops on $v$}}, then neither player has a valid attack and the game ends in a draw with both players having claim to two territories each. 

\item If Raj \textbf{{\color{IndianRed}plays more than $\lceil \nicefrac{T}{2} \rceil$ troops on $v$}}, then Raj can attack $u$ and Lata can attack $y$, and the game ends in a draw with both players having claim to one territory each.
\end{itemize}
Since Raj already has a drawing strategy for any matching by~\Cref{lem:raj-draw-matching}, the claim follows.
\end{proof}

\begin{lemma}
    \label{lem:3edges9troops-win}
    Assume Lata and Raj have~$T$ troops each. If~$G$ is a disjoint union of three edges and~$T \geqslant 9$, then Raj has a winning strategy.
\end{lemma}

\begin{proof} Let the edges be $\{(w,p),(x,q),(y,r)\}$. The strategy we describe will always involve Raj playing on the other endpoint of a partially played edge. Therefore, we assume WLOG that Lata places troops on the vertices $w$, $x$ and $y$ on her first three turns in that order and to the extent that troops are available, and Raj places troops on the vertices $p$, $q$ and $r$ on his first three turns and to the extent that troops are available. If there are troops remaining and empty locations available after the first three turns for either player, then they play the remaining locations in lexicographic order. We emphasize that these conventions are only a matter of making discussion convenient and they can be employed without loss of generality. 

We now describe Raj's strategy. If Lata places five or more troops on $w$ then Raj places only one troop on $p$. We refer to this as a \emph{scary move} for Raj. On the other hand, if Lata places at most four troops on $w$ then  Raj places $f_1(w)+1$ troops on $p$. We refer to this as a \emph{triumphant move} for Raj. 

First, suppose Raj made a scary move on his first turn. Then he responds on $q$ and $r$ with $f_1(x)+1$ and $f_1(y)+1$ troops respectively. It is straightforward to check that this is always feasible, since a scary move leaves Raj with at least four troops more than Lata to play with. It is also easy to see that in these scenarios Raj wins, since he either has two valid attacks, or he has one valid attack and $f_1(y) = 0$, or he has no valid attacks and $f_1(x) = f_1(y) = 0$. In all cases, even after Lata attacks $p$, Raj claims two territories from the remaining edges after the attack phase, so he has one territory more than Lata at the end of the game. 

Now, suppose Raj made a triumphant move on his first turn. Now we have the following cases. We note that our description is exhaustive assuming Lata plays all her troops. The scenario where she plays fewer than nine troops across her three turns is implicit and can also be easily handled effectively with strategy. In the interest of space, we also omit symmetric scenarios and note that $x/y$ and $q/r$ in all the cases below are interchangeable. 

\begin{itemize}
    \item $f_1(w) = 4$. Lata has five troops remaining after her first move and Raj has four.
        \begin{enumerate}
            \item If {\color{IndianRed}$f_1(x) = 5$} and {\color{SeaGreen}$f_1(y) = 0$}, then {\color{IndianRed}$f_2(q) = 1$} and {\color{SeaGreen}$f_2(r) = 1$}.
            \item If {\color{IndianRed}$f_1(x) = 4$} and {\color{SeaGreen}$f_1(y) = 1$}, then {\color{IndianRed}$f_2(q) = 1$} and  {\color{SeaGreen}$f_2(r) = 2$}. 
            \item If {\color{IndianRed}$f_1(x) = 3$} and {\color{SeaGreen}$f_1(y) = 2$}, then {\color{IndianRed}$f_2(q) = 1$} and  {\color{SeaGreen}$f_2(r) = 3$}. 
        \end{enumerate}
    \item $f_1(w) = 3$. Lata has six troops remaining after her first move and Raj has five.    
    \begin{enumerate}
        \item If {\color{IndianRed}$f_1(x) = 6$} and {\color{SeaGreen}$f_1(y) = 0$}, then {\color{IndianRed}$f_2(q) = 1$} and {\color{SeaGreen}$f_2(r) = 1$}.
        \item If {\color{IndianRed}$f_1(x) = 5$} and {\color{SeaGreen}$f_1(y) = 1$}, then {\color{IndianRed}$f_2(q) = 1$} and  {\color{SeaGreen}$f_2(r) = 2$}. 
        \item If {\color{IndianRed}$f_1(x) = 4$} and {\color{SeaGreen}$f_1(y) = 2$}, then {\color{IndianRed}$f_2(q) = 1$} and  {\color{SeaGreen}$f_2(r) = 3$}.
        \item If {\color{SeaGreen}$f_1(x) = 3$} and {\color{IndianRed}$f_1(y) = 3$}, then {\color{SeaGreen}$f_2(q) = 4$} and  {\color{IndianRed}$f_2(r) = 1$}. 
    \end{enumerate}

    \item $f_1(w) = 2$. Lata has seven troops remaining after her first move and Raj has six. 
    \begin{enumerate}
        \item If {\color{IndianRed}$f_1(x) = 7$} and {\color{SeaGreen}$f_1(y) = 0$}, then {\color{IndianRed}$f_2(q) = 1$} and {\color{SeaGreen}$f_2(r) = 1$}.
        \item If {\color{IndianRed}$f_1(x) = 6$} and {\color{SeaGreen}$f_1(y) = 1$}, then {\color{IndianRed}$f_2(q) = 1$} and {\color{SeaGreen} $f_2(r) = 2$}. 
        \item If {\color{IndianRed}$f_1(x) = 5$} and {\color{SeaGreen}$f_1(y) = 2$}, then {\color{IndianRed}$f_2(q) = 1$} and  {\color{SeaGreen}$f_2(r) = 3$}.
        \item If {\color{IndianRed}$f_1(x) = 4$} and {\color{SeaGreen}$f_1(y) = 3$}, then {\color{IndianRed}$f_2(q) = 1$} and  {\color{SeaGreen}$f_2(r) = 4$}. 
    \end{enumerate}
    \item $f_1(w) = 1$. Lata has eight troops remaining after her first move and Raj has seven.
    \begin{enumerate}
        \item If {\color{IndianRed}$f_1(x) = 8$} and {\color{SeaGreen}$f_1(y) = 0$}, then {\color{IndianRed}$f_2(q) = 1$} and {\color{SeaGreen}$f_2(r) = 1$}.
        \item If {\color{IndianRed}$f_1(x) = 7$} and {\color{SeaGreen}$f_1(y) = 1$}, then {\color{IndianRed}$f_2(q) = 1$} and  {\color{SeaGreen}$f_2(r) = 2$}. 
        \item If {\color{IndianRed}$f_1(x) = 6$} and {\color{SeaGreen}$f_1(y) = 2$}, then {\color{IndianRed}$f_2(q) = 1$} and  {\color{SeaGreen}$f_2(r) = 3$}.
        \item If {\color{IndianRed}$f_1(x) = 5$} and {\color{SeaGreen}$f_1(y) = 3$}, then {\color{IndianRed}$f_2(q) = 1$} and  {\color{SeaGreen}$f_2(r) = 4$}. 
        \item If {\color{SeaGreen}$f_1(x) = 4$} and {\color{IndianRed}$f_1(y) = 4$}, then {\color{SeaGreen}$f_2(q) = 5$} and  {\color{IndianRed}$f_2(r) = 1$}. 
    \end{enumerate}
\end{itemize}
Observe that in all scenarios, both Raj and Lata have one valid attack each, which combined with the attack that Raj has from his (triumphant) first move, leads to a win for Raj. This concludes our argument.
\end{proof}

\begin{lemma}
    \label{lem:4edges10troops-win}
    Assume Lata and Raj have~$T$ troops each. If~$G$ is a disjoint union of four edges and~$T \geqslant 10$, then Raj has a winning strategy.
\end{lemma}

\begin{proof} 
Let the edges be $\{(w,p),(x,q),(y,r),(z,s)\}$. The strategy we describe will always involve Raj playing on the other endpoint of a partially played edge. Therefore, we assume WLOG that Lata places troops on the vertices $w$, $x$, $y$ and $z$ on her first four turns in that order and to the extent that troops are available, and Raj places troops on the vertices $p$, $q$, $r$ and $s$ on his first four turns and to the extent that troops are available. The number of troops played by Lata in her turn on the vertex $i$ is denoted by $f_1(i)$.

Similar to the proof of \Cref{lem:3edges9troops-win}, if Lata plays five or more troops on $w$, then Raj plays only one troop on $p$. We refer to this as a \emph{scary move} for Raj. On the other hand, if Lata places at most four troops on $w$ then  Raj places $f_1(w)+1$ troops on $p$. We refer to this as a \emph{triumphant move} for Raj. 

First, suppose Raj made a scary move on his first turn. Then he responds on $q$, $r$ and $s$ with $f_1(x)+1$, $f_1(y)+1$ and $f_1(z)+1$ troops respectively. It is straightforward to check that this is always feasible, since a scary move leaves Raj with at least four troops more than Lata to play with. It is also easy to see that in these scenarios Raj wins, since he either has three valid attacks, or he has two valid attack and $f_1(z) = 0$, or he has one valid attack and $f_1(x) = f_1(y) =0$; or he has no valid attack and $f_1(x)=f_1(y)=f_1(z)=0$. In all cases, even after Lata attacks $p$, Raj claims three territories from the remaining edges after the attack phase, so he has two territories more than Lata at the end of the game. 

Now, suppose Raj made a triumphant move on his first turn. Now we have the following cases. We note that our description is exhaustive assuming Lata plays all her troops. The scenario where she plays fewer than nine troops across her four turns is implicit and can also be easily handled effectively with strategy. In the interest of space, we also omit symmetric scenarios and note that $y/z$ and $r/s$ in all the cases below are interchangeable. 

\begin{itemize}
    \item $f_1(w) = 4$. Lata has six troops remaining after her first move and Raj has five.
        \begin{enumerate}
           \item If {\color{IndianRed}$f_1(x) = 6$} and {\color{SeaGreen}$f_1(y)=f_1(z) = 0$}, then {\color{IndianRed}$f_2(q) = 1$} and {\color{SeaGreen}$f_2(r)=f_2(s) = 1$}.
            \item If {\color{IndianRed}$f_1(x) = 5$} ,{\color{SeaGreen}$f_1(y) = 1$} , and {\color{SeaGreen}$f_1(z)=0$}; then {\color{IndianRed}$f_2(q) = 1$} and {\color{SeaGreen}$f_2(r) = 2$} , {\color{SeaGreen}$f_2(s)=1$}.
            \item If {\color{IndianRed}$f_1(x) = 4$} ,{\color{SeaGreen}$f_1(y) = 2$}, and {\color{SeaGreen}$f_2(z)=0$}, then {\color{IndianRed}$f_2(q) = 1$} and  {\color{SeaGreen}$f_2(r) = 3$} , {\color{SeaGreen}$f_2(s)=1$}. 
             \item If {\color{IndianRed}$f_1(x) = 4$} ,{\color{SeaGreen}$f_1(y) = 1$,} and
             {\color{SeaGreen}$f_2(z)=1$}, then {\color{IndianRed}$f_2(q) = 1$} , {\color{SeaGreen}$f_2(r) = 2$} and {\color{SeaGreen}$f_2(s)=2$}. 
             \item If {\color{IndianRed}$f_1(x) = 3$} ,{\color{SeaGreen}$f_1(y) = 3$,} and
             {\color{DodgerBlue}$f_2(z)=0$}, then {\color{IndianRed}$f_2(q) = 1$} , {\color{SeaGreen}$f_2(r) = 4$} and {\color{DodgerBlue}$f_2(s)=0$}. 
              \item If {\color{IndianRed}$f_1(x) = 3$} ,{\color{SeaGreen}$f_1(y) = 2$,} and
             {\color{DodgerBlue}$f_2(z)=1$}, then {\color{IndianRed}$f_2(q) = 1$} , {\color{SeaGreen}$f_2(r) = 3$} and {\color{DodgerBlue}$f_2(s)=1$}. 
             \item If {\color{SeaGreen}$f_1(x) = 2$} ,{\color{IndianRed}$f_1(y) = 4$,} and
             {\color{SeaGreen}$f_2(s)=0$}, then {\color{SeaGreen}$f_2(q) = 3$} , {\color{IndianRed}$f_2(r) = 1$} and {\color{SeaGreen}$f_2(s)=1$}. 
             \item If {\color{SeaGreen}$f_1(x) = 2$} ,{\color{IndianRed}$f_1(y) = 3$,} and
             {\color{DodgerBlue}$f_2(s)=1$}, then {\color{SeaGreen}$f_2(q) = 3$} , {\color{IndianRed}$f_2(r) = 1$} and {\color{DodgerBlue}$f_2(s)=1$}.
             \item If {\color{SeaGreen}$f_1(x) = 2$} ,{\color{DodgerBlue}$f_1(y) = 2$,} and
             {\color{IndianRed}$f_2(z)=2$}, then {\color{SeaGreen}$f_2(q) = 3$} , {\color{DodgerBlue}$f_2(r) = 2$} and {\color{IndianRed}$f_2(s)=0$}.
             \item If {\color{SeaGreen}$f_1(x) = 1$} ,{\color{IndianRed}$f_1(y) = 5$,} and
             {\color{SeaGreen}$f_2(z)=0$}, then {\color{SeaGreen}$f_2(q) = 2$} , {\color{IndianRed}$f_2(r) = 1$} and {\color{SeaGreen}$f_2(s)=1$}.
             \item If {\color{SeaGreen}$f_1(x) = 1$} ,{\color{IndianRed}$f_1(y) = 4$,} and
             {\color{SeaGreen}$f_2(z)=1$}, then {\color{SeaGreen}$f_2(q) = 2$} , {\color{IndianRed}$f_2(r) = 1$} and {\color{SeaGreen}$f_2(s)=2$}.
             \item If {\color{SeaGreen}$f_1(x) = 1$} ,{\color{IndianRed}$f_1(y) = 3$,} and
             {\color{DodgerBlue}$f_2(z)=2$}, then {\color{SeaGreen}$f_2(q) = 2$} , {\color{IndianRed}$f_2(r) = 1$} and {\color{DodgerBlue}$f_2(s)=2$}.    
        \end{enumerate}
    \item $f_1(w) = 3$. Lata has seven troops remaining after her first move and Raj has six.    
    \begin{enumerate}
        \item If {\color{IndianRed}$f_1(x) = 7$} and {\color{SeaGreen}$f_1(y)=f_1(z) = 0$}, then {\color{IndianRed}$f_2(q) = 1$} and {\color{SeaGreen}$f_2(r)=f_2(s) = 1$}.
            \item If {\color{IndianRed}$f_1(x) = 6$} ,{\color{SeaGreen}$f_1(y) = 1$} , and {\color{SeaGreen}$f_1(z)=0$}; then {\color{IndianRed}$f_2(q) = 1$} and {\color{SeaGreen}$f_2(r) = 2$} , {\color{SeaGreen}$f_2(s)=1$}.
            \item If {\color{IndianRed}$f_1(x) = 5$} ,{\color{SeaGreen}$f_1(y) = 1$} , and {\color{SeaGreen}$f_1(z)=1$}; then {\color{IndianRed}$f_2(q) = 1$} and {\color{SeaGreen}$f_2(r) = 2$} , {\color{SeaGreen}$f_2(s)=2$}.
            \item If {\color{IndianRed}$f_1(x) = 4$} ,{\color{SeaGreen}$f_1(y) = 3$}, and {\color{SeaGreen}$f_2(z)=0$}, then {\color{IndianRed}$f_2(q) = 1$} and  {\color{SeaGreen}$f_2(r) = 4$} , {\color{SeaGreen}$f_2(s)=1$}. 
             \item If {\color{IndianRed}$f_1(x) = 4$} ,{\color{SeaGreen}$f_1(y) = 2$,} and
             {\color{SeaGreen}$f_2(z)=1$}, then {\color{IndianRed}$f_2(q) = 1$} , {\color{SeaGreen}$f_2(r) = 3$} and {\color{SeaGreen}$f_2(s)=2$}. 
             \item If {\color{SeaGreen}$f_1(x) = 3$} ,{\color{IndianRed}$f_1(y) = 4$,} and
             {\color{SeaGreen}$f_2(z)=0$}, then {\color{SeaGreen}$f_2(q) = 4$} , {\color{IndianRed}$f_2(r) = 1$} and {\color{SeaGreen}$f_2(s)=1$}. 
             \item If {\color{SeaGreen}$f_1(x) = 3$} ,{\color{IndianRed}$f_1(y) = 3$,} and
             {\color{DodgerBlue}$f_2(z)=1$}, then {\color{SeaGreen}$f_2(q) = 4$} , {\color{IndianRed}$f_2(r) = 1$} and {\color{DodgerBlue}$f_2(s)=1$}. 
             \item If {\color{SeaGreen}$f_1(x) = 3$} ,{\color{DodgerBlue}$f_1(y) = 2$,} and
             {\color{IndianRed}$f_2(z)=2$}, then {\color{SeaGreen}$f_2(q) = 4$} , {\color{DodgerBlue}$f_2(r) = 2$} and {\color{IndianRed}$f_2(s)=0$}. 
             \item If {\color{SeaGreen}$f_1(x) = 2$} ,{\color{IndianRed}$f_1(y) = 5$,} and
             {\color{SeaGreen}$f_2(z)=0$}, then {\color{SeaGreen}$f_2(q) = 3$} , {\color{IndianRed}$f_2(r) = 1$} and {\color{SeaGreen}$f_2(s)=1$}. 
             \item If {\color{SeaGreen}$f_1(x) = 2$} ,{\color{IndianRed}$f_1(y) = 4$,} and
             {\color{SeaGreen}$f_2(z)=1$}, then {\color{SeaGreen}$f_2(q) = 3$} , {\color{IndianRed}$f_2(r) = 1$} and {\color{SeaGreen}$f_2(s)=2$}.
             \item If {\color{SeaGreen}$f_1(x) = 2$} ,{\color{IndianRed}$f_1(y) = 3$,} and
             {\color{DodgerBlue}$f_2(z)=2$}, then {\color{SeaGreen}$f_2(q) = 3$} , {\color{IndianRed}$f_2(r) = 1$} and {\color{DodgerBlue}$f_2(s)=2$}.
             \item If {\color{SeaGreen}$f_1(x) = 1$} ,{\color{IndianRed}$f_1(y) = 6$,} and
             {\color{SeaGreen}$f_2(z)=0$}, then {\color{SeaGreen}$f_2(q) = 2$} , {\color{IndianRed}$f_2(r) = 1$} and {\color{SeaGreen}$f_2(s)=1$}.
             \item If {\color{SeaGreen}$f_1(x) = 1$} ,{\color{IndianRed}$f_1(y) = 5$,} and
             {\color{SeaGreen}$f_2(z)=1$}, then {\color{SeaGreen}$f_2(q) = 2$} , {\color{IndianRed}$f_2(r) = 1$} and {\color{SeaGreen}$f_2(s)=2$}.
             \item If {\color{SeaGreen}$f_1(x) = 1$} ,{\color{IndianRed}$f_1(y) = 4$,} and
             {\color{DodgerBlue}$f_2(z)=2$}, then {\color{SeaGreen}$f_2(q) = 2$} , {\color{IndianRed}$f_2(r) = 1$} and {\color{DodgerBlue}$f_2(s)=2$}. 
             \item If {\color{SeaGreen}$f_1(x) = 1$} ,{\color{SeaGreen}$f_1(y) = 3$,} and
             {\color{IndianRed}$f_2(z)=3$}, then {\color{SeaGreen}$f_2(q) = 2$} , {\color{SeaGreen}$f_2(r) = 3$} and {\color{IndianRed}$f_2(s)=0$}. 
    \end{enumerate}

    \item $f_1(w) = 2$. Lata has eight troops remaining after her first move and Raj has seven. 
    \begin{enumerate}
            \item If {\color{IndianRed}$f_1(x) = 8$} and {\color{SeaGreen}$f_1(y)=f_1(z) = 0$}, then {\color{IndianRed}$f_2(q) = 1$} and {\color{SeaGreen}$f_2(r)=f_2(s) = 1$}.
            \item If {\color{IndianRed}$f_1(x) = 7$} ,{\color{SeaGreen}$f_1(y) = 1$} , and {\color{SeaGreen}$f_1(z)=0$}; then {\color{IndianRed}$f_2(q) = 1$} and {\color{SeaGreen}$f_2(r) = 2$} , {\color{SeaGreen}$f_2(s)=1$}.
            \item If {\color{IndianRed}$f_1(x) = 6$} ,{\color{SeaGreen}$f_1(y) = 2$} , and {\color{SeaGreen}$f_1(z)=0$}; then {\color{IndianRed}$f_2(q) = 1$} and {\color{SeaGreen}$f_2(r) = 3$} , {\color{SeaGreen}$f_2(s)=1$}.
             \item If {\color{IndianRed}$f_1(x) = 6$} ,{\color{SeaGreen}$f_1(y) = 1$} , and {\color{SeaGreen}$f_1(z)=1$}; then {\color{IndianRed}$f_2(q) = 1$} and {\color{SeaGreen}$f_2(r) = 2$} , {\color{SeaGreen}$f_2(s)=2$}.
            \item If {\color{IndianRed}$f_1(x) = 5$} ,{\color{SeaGreen}$f_1(y) = 3$} , and {\color{SeaGreen}$f_1(z)=0$}; then {\color{IndianRed}$f_2(q) = 1$} and {\color{SeaGreen}$f_2(r) = 4$} , {\color{SeaGreen}$f_2(s)=1$}. 
            \item If {\color{IndianRed}$f_1(x) = 5$} ,{\color{SeaGreen}$f_1(y) = 2$} , and {\color{SeaGreen}$f_1(z)=1$}; then {\color{IndianRed}$f_2(q) = 1$} and {\color{SeaGreen}$f_2(r) = 3$} , {\color{SeaGreen}$f_2(s)=2$}. 
            \item If {\color{IndianRed}$f_1(x) = 4$} ,{\color{SeaGreen}$f_1(y) = 4$}, and {\color{SeaGreen}$f_2(z)=0$}, then {\color{IndianRed}$f_2(q) = 1$} and  {\color{SeaGreen}$f_2(r) = 5$} , {\color{SeaGreen}$f_2(s)=1$}. 
            \item If {\color{IndianRed}$f_1(x) = 4$} ,{\color{SeaGreen}$f_1(y) = 3$,} and
            {\color{SeaGreen}$f_2(z)=1$}, then {\color{IndianRed}$f_2(q) = 1$} , {\color{SeaGreen}$f_2(r) = 4$} and {\color{SeaGreen}$f_2(s)=2$}. 
            \item If {\color{IndianRed}$f_1(x) = 4$} ,{\color{SeaGreen}$f_1(y) = 2$,} and
            {\color{SeaGreen}$f_2(z)=2$}, then {\color{IndianRed}$f_2(q) = 1$} , {\color{SeaGreen}$f_2(r) = 3$} and {\color{SeaGreen}$f_2(s)=3$}. 
            \item If {\color{SeaGreen}$f_1(x) = 3$} ,{\color{IndianRed}$f_1(y) = 5$,} and
            {\color{SeaGreen}$f_2(z)=0$}, then {\color{SeaGreen}$f_2(q) = 4$} , {\color{IndianRed}$f_2(r) = 1$} and {\color{SeaGreen}$f_2(s)=1$}. 
            \item If {\color{SeaGreen}$f_1(x) = 3$} ,{\color{IndianRed}$f_1(y) = 4$,} and
            {\color{SeaGreen}$f_2(s)=1$}, then {\color{SeaGreen}$f_2(q) = 4$} , {\color{IndianRed}$f_2(r) = 1$} and {\color{SeaGreen}$f_2(s)=2$}. 
            \item If {\color{SeaGreen}$f_1(x) = 3$} ,{\color{IndianRed}$f_1(y) = 3$,} and
             {\color{DodgerBlue}$f_2(z)=2$}, then {\color{SeaGreen}$f_2(q) = 4$} , {\color{IndianRed}$f_2(r) = 1$} and {\color{DodgerBlue}$f_2(s)=2$}. 
             \item If {\color{SeaGreen}$f_1(x) = 2$} ,{\color{IndianRed}$f_1(y) = 6$,} and
             {\color{SeaGreen}$f_2(z)=0$}, then {\color{SeaGreen}$f_2(q) = 3$} , {\color{IndianRed}$f_2(r) = 1$} and {\color{SeaGreen}$f_2(s)=1$}. 
             \item If {\color{SeaGreen}$f_1(x) = 2$} ,{\color{IndianRed}$f_1(y) = 5$,} and
             {\color{SeaGreen}$f_2(z)=1$}, then {\color{SeaGreen}$f_2(q) = 3$} , {\color{IndianRed}$f_2(r) = 1$} and {\color{SeaGreen}$f_2(s)=2$}.
             \item If {\color{SeaGreen}$f_1(x) = 2$} ,{\color{IndianRed}$f_1(y) = 4$,} and
             {\color{SeaGreen}$f_2(z)=2$}, then {\color{SeaGreen}$f_2(q) = 3$} , {\color{IndianRed}$f_2(r) = 1$} and {\color{SeaGreen}$f_2(s)=3$}.
             \item If {\color{SeaGreen}$f_1(x) = 2$} ,{\color{IndianRed}$f_1(y) = 3$,} and
             {\color{DodgerBlue}$f_2(z)=3$}, then {\color{SeaGreen}$f_2(q) = 3$} , {\color{IndianRed}$f_2(r) = 1$} and {\color{DodgerBlue}$f_2(s)=3$}.
             \item If {\color{SeaGreen}$f_1(x) = 1$} ,{\color{IndianRed}$f_1(y) = 7$,} and
             {\color{SeaGreen}$f_2(z)=0$}, then {\color{SeaGreen}$f_2(q) = 2$} , {\color{IndianRed}$f_2(r) = 1$} and {\color{SeaGreen}$f_2(s)=1$}.
             \item If {\color{SeaGreen}$f_1(x) = 1$} ,{\color{IndianRed}$f_1(y) = 6$,} and
             {\color{SeaGreen}$f_2(z)=1$}, then {\color{SeaGreen}$f_2(q) = 2$} , {\color{IndianRed}$f_2(r) = 1$} and {\color{SeaGreen}$f_2(s)=2$}.
             \item If {\color{SeaGreen}$f_1(x) = 1$} ,{\color{IndianRed}$f_1(y) = 5$,} and
             {\color{SeaGreen}$f_2(z)=2$}, then {\color{SeaGreen}$f_2(q) = 2$} , {\color{IndianRed}$f_2(r) = 1$} and {\color{SeaGreen}$f_2(s)=3$}. 
             \item If {\color{SeaGreen}$f_1(x) = 1$} ,{\color{SeaGreen}$f_1(y) = 4$,} and
             {\color{IndianRed}$f_2(s)=3$}, then {\color{SeaGreen}$f_2(q) = 2$} , {\color{SeaGreen}$f_2(r) = 5$} and {\color{IndianRed}$f_2(s)=0$}. 
             
    \end{enumerate}
    \item $f_1(w) = 1$. Lata has nine troops remaining after her first move and Raj has eight.
    \begin{enumerate}
        \item If {\color{IndianRed}$f_1(x) = 9$} and {\color{SeaGreen}$f_1(y)=f_1(z) = 0$}, then {\color{IndianRed}$f_2(q) = 1$} and {\color{SeaGreen}$f_2(r)=f_2(s) = 1$}.
            \item If {\color{IndianRed}$f_1(x) = 8$} ,{\color{SeaGreen}$f_1(y) = 1$} , and {\color{SeaGreen}$f_1(z)=0$}; then {\color{IndianRed}$f_2(q) = 1$} and {\color{SeaGreen}$f_2(r) = 2$} , {\color{SeaGreen}$f_2(s)=1$}.
            \item If {\color{IndianRed}$f_1(x) = 7$} ,{\color{SeaGreen}$f_1(y) = 2$} , and {\color{SeaGreen}$f_1(z)=0$}; then {\color{IndianRed}$f_2(q) = 1$} and {\color{SeaGreen}$f_2(r) = 3$} , {\color{SeaGreen}$f_2(s)=1$}.
             \item If {\color{IndianRed}$f_1(x) = 7$} ,{\color{SeaGreen}$f_1(y) = 1$} , and {\color{SeaGreen}$f_1(z)=1$}; then {\color{IndianRed}$f_2(q) = 1$} and {\color{SeaGreen}$f_2(r) = 2$} , {\color{SeaGreen}$f_2(s)=2$}.
            \item If {\color{IndianRed}$f_1(x) = 6$} ,{\color{SeaGreen}$f_1(y) = 3$} , and {\color{SeaGreen}$f_1(z)=0$}; then {\color{IndianRed}$f_2(q) = 1$} and {\color{SeaGreen}$f_2(r) = 4$} , {\color{SeaGreen}$f_2(s)=1$}. 
            \item If {\color{IndianRed}$f_1(x) = 6$} ,{\color{SeaGreen}$f_1(y) = 2$} , and {\color{SeaGreen}$f_1(z)=1$}; then {\color{IndianRed}$f_2(q) = 1$} and {\color{SeaGreen}$f_2(r) = 3$} , {\color{SeaGreen}$f_2(s)=2$}. 
            \item If {\color{IndianRed}$f_1(x) = 5$} ,{\color{SeaGreen}$f_1(y) = 4$}, and {\color{SeaGreen}$f_2(z)=0$}, then {\color{IndianRed}$f_2(q) = 1$} and  {\color{SeaGreen}$f_2(r) = 5$} , {\color{SeaGreen}$f_2(s)=1$}. 
            \item If {\color{IndianRed}$f_1(x) = 5$} ,{\color{SeaGreen}$f_1(y) = 3$,} and
            {\color{SeaGreen}$f_2(z)=1$}, then {\color{IndianRed}$f_2(q) = 1$} , {\color{SeaGreen}$f_2(r) = 4$} and {\color{SeaGreen}$f_2(s)=2$}. 
            \item If {\color{IndianRed}$f_1(x) = 5$} ,{\color{SeaGreen}$f_1(y) = 2$,} and
            {\color{SeaGreen}$f_2(z)=2$}, then {\color{IndianRed}$f_2(q) = 1$} , {\color{SeaGreen}$f_2(r) = 3$} and {\color{SeaGreen}$f_2(s)=3$}. 
            \item If {\color{SeaGreen}$f_1(x) = 4$} ,{\color{IndianRed}$f_1(y) = 5$,} and
            {\color{SeaGreen}$f_2(z)=0$}, then {\color{SeaGreen}$f_2(q) = 5$} , {\color{IndianRed}$f_2(r) = 1$} and {\color{SeaGreen}$f_2(s)=1$}. 
            \item If {\color{SeaGreen}$f_1(x) = 4$} ,{\color{IndianRed}$f_1(y) = 4$,} and
            {\color{SeaGreen}$f_2(z)=1$}, then {\color{SeaGreen}$f_2(q) = 5$} , {\color{IndianRed}$f_2(r) = 1$} and {\color{SeaGreen}$f_2(s)=2$}. 
            \item If {\color{SeaGreen}$f_1(x) = 4$} ,{\color{IndianRed}$f_1(y) = 3$,} and
             {\color{DodgerBlue}$f_2(z)=2$}, then {\color{SeaGreen}$f_2(q) = 5$} , {\color{IndianRed}$f_2(r) = 1$} and {\color{DodgerBlue}$f_2(s)=2$}. 
             \item If {\color{SeaGreen}$f_1(x) = 3$} ,{\color{IndianRed}$f_1(y) = 6$,} and
             {\color{SeaGreen}$f_2(z)=0$}, then {\color{SeaGreen}$f_2(q) = 4$} , {\color{IndianRed}$f_2(r) = 1$} and {\color{SeaGreen}$f_2(s)=1$}. 
             \item If {\color{SeaGreen}$f_1(x) = 3$} ,{\color{IndianRed}$f_1(y) = 5$,} and
             {\color{SeaGreen}$f_2(z)=1$}, then {\color{SeaGreen}$f_2(q) = 3$} , {\color{IndianRed}$f_2(r) = 1$} and {\color{SeaGreen}$f_2(s)=2$}.
             \item If {\color{SeaGreen}$f_1(x) = 3$} ,{\color{IndianRed}$f_1(y) = 4$,} and
             {\color{SeaGreen}$f_2(z)=2$}, then {\color{SeaGreen}$f_2(q) = 3$} , {\color{IndianRed}$f_2(r) = 1$} and {\color{SeaGreen}$f_2(s)=3$}.
             \item If {\color{SeaGreen}$f_1(x) = 3$} ,{\color{IndianRed}$f_1(y) = 3$,} and
             {\color{DodgerBlue}$f_2(z)=3$}, then {\color{SeaGreen}$f_2(q) = 4$} , {\color{IndianRed}$f_2(r) = 1$} and {\color{DodgerBlue}$f_2(s)=3$}.
             \item If {\color{SeaGreen}$f_1(x) = 2$} ,{\color{IndianRed}$f_1(y) = 7$,} and
             {\color{SeaGreen}$f_2(z)=0$}, then {\color{SeaGreen}$f_2(q) = 3$} , {\color{IndianRed}$f_2(r) = 1$} and {\color{SeaGreen}$f_2(s)=1$}.
             \item If {\color{SeaGreen}$f_1(x) = 2$} ,{\color{IndianRed}$f_1(y) = 6$,} and
             {\color{SeaGreen}$f_2(z)=1$}, then {\color{SeaGreen}$f_2(q) = 3$} , {\color{IndianRed}$f_2(r) = 1$} and {\color{SeaGreen}$f_2(s)=2$}.
             \item If {\color{SeaGreen}$f_1(x) = 2$} ,{\color{IndianRed}$f_1(y) = 5$,} and
             {\color{SeaGreen}$f_2(z)=2$}, then {\color{SeaGreen}$f_2(q) = 3$} , {\color{IndianRed}$f_2(r) = 1$} and {\color{SeaGreen}$f_2(s)=3$}. 
             \item If {\color{SeaGreen}$f_1(x) = 2$} ,{\color{SeaGreen}$f_1(y) = 4$,} and
             {\color{IndianRed}$f_2(z)=3$}, then {\color{SeaGreen}$f_2(q) = 3$} , {\color{SeaGreen}$f_2(r) = 5$} and {\color{IndianRed}$f_2(s)=0$}. 
             \item If {\color{SeaGreen}$f_1(x) = 1$} ,{\color{IndianRed}$f_1(y) = 8$,} and
             {\color{SeaGreen}$f_2(z)=0$}, then {\color{SeaGreen}$f_2(q) = 2$} , {\color{IndianRed}$f_2(r) = 1$} and {\color{SeaGreen}$f_2(s)=1$}.
             \item If {\color{SeaGreen}$f_1(x) = 1$} ,{\color{IndianRed}$f_1(y) = 7$,} and
             {\color{SeaGreen}$f_2(z)=1$}, then {\color{SeaGreen}$f_2(q) = 2$} , {\color{IndianRed}$f_2(r) = 1$} and {\color{SeaGreen}$f_2(s)=2$}.
             \item If {\color{SeaGreen}$f_1(x) = 1$} ,{\color{IndianRed}$f_1(y) = 6$,} and
             {\color{SeaGreen}$f_2(z)=2$}, then {\color{SeaGreen}$f_2(q) = 2$} , {\color{IndianRed}$f_2(r) = 1$} and {\color{SeaGreen}$f_2(s)=3$}. 
             \item If {\color{SeaGreen}$f_1(x) = 1$} ,{\color{IndianRed}$f_1(y) = 5$,} and
             {\color{SeaGreen}$f_2(z)=3$}, then {\color{SeaGreen}$f_2(q) = 2$} , {\color{IndianRed}$f_2(r) = 1$} and {\color{SeaGreen}$f_2(s)=4$}. 
             \item If {\color{SeaGreen}$f_1(x) = 1$} ,{\color{DodgerBlue}$f_1(y) = 4$,} and
             {\color{IndianRed}$f_2(s)=4$}, then {\color{SeaGreen}$f_2(q) = 2$} , {\color{IndianRed}$f_2(r) = 4$} and {\color{DodgerBlue}$f_2(s)=0$}. 
    \end{enumerate}
\end{itemize}
It can be seen that for all possible plays of Lata, Raj either draws or wins on the last three edges. Since Raj has already made a triumphant move, Raj has an extra attack on the edge $wp$. Thus Raj wins.
\end{proof}

\begin{lemma}
    \label{lem:4edges10troops-strongwin}
    Assume Lata and Raj have~$T_R$ and~$T_L$ troops respectively. If~$G$ is a disjoint union of four edges,~$T_R \geqslant 10$ and~$T_L \leqslant 9$, then Raj has a strongly winning strategy.
\end{lemma}

\begin{proof} We adapt the argument we made in the proof of~\Cref{lem:4edges10troops-win} and assume that Lata has nine troops and Raj has ten: the scenarios when Lata has fewer troops and/or Raj has more can be easily argued in a similar fashion.

When Raj makes a scary first move he already has secured a strong win. When $f_1(w) = 1, 2$ or $3$, Lata now has eight, seven, and six troops remaining respectively, while Raj still has nine, eight, and seven troops. So we appeal to the cases in the proof of~\Cref{lem:4edges10troops-win} for when Lata has eight, seven, and six troops remaining, but modify Raj's response to use one extra troop whenever the response was originally a draw. This ensures a strong win across all cases. When $f_1(w) = 4$, Lata has five troops remaining and Raj has six. In this case it is easily verified that Raj can respond by playing $f_2(p) = 1$, giving up on the first edge, and then responding on $q$, $r$ and $s$ with $f_1(x)+1$, $f_1(y)+1$ and $f_1(z)+1$ troops: this is again a strong win.
\end{proof}

\begin{theorem}
    Let~$G$ be a disjoint union of~$N$ edges with~$N \geqslant 4$. Then:
    \begin{enumerate}
        \item[(a)] If both players have~$T$ troops each and~$T \geqslant N+6$, then Raj has a winning strategy.
        \item[(b)] If Raj has~$T_R \geqslant N+6$ troops and Lata has~$T_L$ troops with~$T_L < T_R$, then Raj has a strongly winning strategy.
    \end{enumerate}
\end{theorem}

\begin{proof}
We proceed by (strong) induction on the number of edges. The base case is when~$N = 4$, and the two claims follow directly from~\Cref{lem:4edges10troops-strongwin,lem:4edges10troops-win}, respectively. Now assume that the claims hold for all~$4 \leqslant M < N$ edges, and let the edges of the graph~$G$ be given by:

$$E(G) := \{e_1 = (u_1,v_1), \cdots, e_i = (u_i,v_i), \cdots, e_N = (u_N,v_N)\}.$$

The strategy employed by Raj will involve playing on the ``other endpoint'' of the edge chosen by Lata, so we will maintain the invariant that after~$r$ rounds, there are~$r$ edges both of whose endpoints are fully occupied, and~$(N-r)$ edges that have not been played on at all. Therefore, without loss of generality and by appropriate renaming, we assume that Lata places~$t_i$ troops on the vertex~$u_i$ on her~$i\textsuperscript{th}$ turn. 

We say that a game state is \emph{balanced} if both Raj and Lata have the same number of troops. Suppose we have a balanced game after~$i$ rounds, where both players have~$k$ troops. We say that a move by Lata that involves placing~$t_i$ troops on~$u_i$ is:

\begin{itemize}
    \item {\color{IndianRed}\textbf{dangerous}} if~$k - t_i < (N-i) + 6$,
    \item {\color{DodgerBlue}\textbf{a cliffhanger}} if~$i = N-3$, in other words, three edges remain after this move, and
    \item {\color{SeaGreen}\textbf{normal}} otherwise.
\end{itemize}

Note that the game is balanced to begin with. Now, the strategy employed by Raj is the following. If the~$i\textsuperscript{th}$ move made by Lata is normal, then Raj responds by placing~$t_i$ troops on~$v_i$. This response maintains the invariant that the game is balanced on all normal moves.

Let~$e_i$ be the first edge on which Lata makes a move that is not normal. This move is either dangerous or a cliffhanger. We argue these cases separately. 

Suppose the move is a cliffhanger. Since all the previous moves were normal, we know that the game is balanced and Raj has at least~$(N-(N-4)+6) = 10$ troops at this stage of the game. Raj then responds to Lata's current and future moves according to the strategy described in the proof of~\Cref{lem:4edges10troops-win}, and wins the game.

Suppose now that the move is dangerous, and WLOG not a cliffhanger. Then Raj responds by placing one troop on~$v_i$. Because this is the first move that is not normal, the game was balanced up to this point --- say with~$k$ troops --- and Raj had at least~$(N-i+1) + 6$ troops, therefore he has at least~$(N-i) + 6$ troops after this move. Also, Lata has~$k - t_i < (N-i) + 6$ troops, since the move was dangerous. Therefore, after this move, on the remaining game with~$(N-i)$ edges, Raj satisfies the induction hypothesis and can play according to a strongly winning strategy. Note that~$(N-i) \geqslant 4$, because the move is not a cliffhanger. Once the game is complete, Raj has at least two territories more than Lata on the game with~$(N-i)$ edges, and has one territory less than Lata on the game with the first~$i$ edges: recall that the first~$(i-1)$ edges were drawn and Lata has a valid attack on the edge~$e_i$.
\end{proof}

\subsection{Cycles}



We now show that both players have drawing strategies when the graph is a cycle of length at most five.

\begin{lemma}
    \label{lem:short-cycles-draw}
    Assume Lata and Raj have~$T$ troops each. If~$G$ is a cycle on at most five vertices, both players have a drawing strategy.
\end{lemma}

\begin{proof}
    We argue the cases for cycles of lengths three, four, and five separately.

\paragraph*{$G$ is a triangle.} It is easily checked that both players can draw the game by placing all their troops on their turn on any available vertex.
    
\paragraph*{$G$ is a $4$-cycle.} Lata can always ensure that they are able to place their troops on two adjacent vertices, no matter how Raj makes their first move. By an argument similar to the proof of~\Cref{lem:2edges-draw}, we see that Lata can draw the game by placing $\lceil \nicefrac{T}{2} \rceil$ troops on one of the vertices and $\lfloor \nicefrac{T}{2} \rfloor$ troops on an available neighbor.

We now argue a draw strategy for Raj. Let the vertices of the $4$-cycle be $\{v_1,v_2,v_3,v_4\}$ and WLOG suppose Lata positions $a$ troops on $v_1$ on their turn. If $a = T$, then Raj responds by placing $T$ troops on $v_3$ and the game ends immediately in a draw. If $a < T$, Raj places $\lceil \nicefrac{T}{2} \rceil$ troops on $v_3$ and $\lfloor \nicefrac{T}{2} \rfloor$ troops on the vertex that remains. Notice that matter where Lata plays on their second move, Raj has their troops on two adjacent vertices. The argument that the game is a draw is now similar to the argument we made for Lata.

\paragraph*{$G$ is a $5$-cycle.}

We consider the~$5$-cycle on the vertices~$\{v_1, v_2, v_3, v_4, v_5\}$.

\textbf{A Draw Strategy for Lata.} Lata begins by placing $\lfloor \nicefrac{T}{2} \rfloor$ troops on $v_1$. 

Suppose Raj responds by placing $a$ troops on $v_2$ (the case for $v_5$ is symmetric). We then have the following cases: 

\begin{itemize}
\item If $a = T$, Lata responds by placing one troop on $v_5$ and their remaining troops on $v_4$. Raj has one attack and Lata has none, but Lata ends with two territories and Raj has one.
\item If $a < T$, Lata places $\lceil \nicefrac{T}{2} \rceil$ troops on $v_3$. Now, suppose Raj places $(T-a)$ troops on $v_4$ for his second move. Lata has the first attack, which she uses to attack $v_2$. Note that this is always a feasible attack since $a < T$, and the game is draw even if Raj has a valid attack remaining. The reasoning for the case when Raj places $(T-a)$ troops on $v_5$ for his second move works out the same way.
\end{itemize}

Now suppose Raj responds by placing $a$ troops on $v_3$ (the case for $v_4$ is symmetric). We then have the following cases: 

We then have the following cases: 

\begin{itemize}
\item If $a = T$, Lata responds by placing her remaining troops on $v_5$ and wins the game: she owns two territories and Raj has one, and neither player has a valid attack.
\item If $a < \lceil \nicefrac{T}{2} \rceil$, Lata places $\lceil \nicefrac{T}{2} \rceil$ troops on $v_2$. Now, suppose Raj places $(T-a)$ troops on $v_4$ for his second move. Lata has the first attack, which she uses to attack $v_3$. Note that this is always a feasible attack since $a < \lceil \nicefrac{T}{2} \rceil$, and Raj has no attacks, so Lata wins. If Raj places $(T-a)$ troops on $v_5$ for his second move, then Lata attacks Raj on $v_3$ on her turn using her troops on $v_2$. Raj then attacks Lata's troops on $v_1$, and the game is drawn with both players owning one territory at the end of the game.
\item If $a = \lceil \nicefrac{T}{2} \rceil$, Lata places $\lceil \nicefrac{T}{2} \rceil$ troops on $v_2$. Now, suppose Raj places his remaining $\lfloor \nicefrac{T}{2} \rfloor$ troops on $v_4$ for his second move. The game then draws with no feasible attacks. On the other hand, if Raj places $\lfloor \nicefrac{T}{2} \rfloor$ troops on $v_5$ for his second move, then Lata attacks Raj on $v_5$ on her turn using her troops on $v_1$, and Raj has no attacks, so Lata wins.
\item If $a > \lceil \nicefrac{T}{2} \rceil$, Lata places $\lceil \nicefrac{T}{2} \rceil$ troops on $v_5$. Now, suppose Raj places his remaining $(T-a)$ troops on either $v_4$ or $v_2$ for his second move. Lata has the first attack, which she uses to attack the troops Raj placed on his second move using her troops on $v_5$ or $v_1$ respectively. Note that this is always a feasible attack since:

$$a > \lceil \nicefrac{T}{2} \rceil \implies T-a < \lfloor \nicefrac{T}{2} \rfloor.$$

After this, it is easy to check that Raj has no attacks, so Lata wins.
\end{itemize}

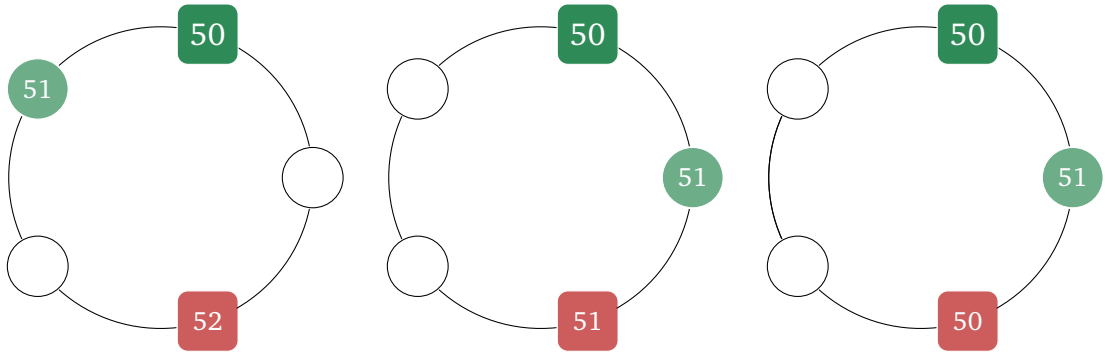
\begin{figure}[t]
    \begin{tikzpicture}[scale=0.5]
    
        \def \n {5}
        \def \radius {4cm}
        \def \margin {12} 
        
        \node[draw, black, circle, fill=white, minimum size=0.8cm] at ({360/5 * (1 - 1)}:\radius) {};
        \draw[-, >=latex] ({360/5 * (1 - 1)+\margin}:\radius) 
            arc ({360/5 * (1 - 1)+\margin}:{360/5 * (1)-\margin}:\radius);
        
        \node[draw, white, rounded corners, rectangle, fill=SeaGreen, minimum size=0.8cm] at ({360/5 * (2 - 1)}:\radius) {\Large{50}};
        \draw[-, >=latex] ({360/5 * (2 - 1)+\margin}:\radius) 
            arc ({360/5 * (2 - 1)+\margin}:{360/5 * (2)-\margin}:\radius);
        
        \node[draw, white, circle, fill=SeaGreen!70, minimum size=0.8cm] at ({360/5 * (3 - 1)}:\radius) {51};
        \draw[-, >=latex] ({360/5 * (3 - 1)+\margin}:\radius) 
            arc ({360/5 * (3 - 1)+\margin}:{360/5 * (3)-\margin}:\radius);
        
        \node[draw, black, circle, fill=white, minimum size=0.8cm] at ({360/5 * (4 - 1)}:\radius) {};
        \draw[-, >=latex] ({360/5 * (4 - 1)+\margin}:\radius) 
            arc ({360/5 * (4 - 1)+\margin}:{360/5 * (4)-\margin}:\radius);
        
        \node[draw, white, rounded corners, rectangle, fill=IndianRed, minimum size=0.8cm] at ({360/5 * (5 - 1)}:\radius) {52};
        \draw[-, >=latex] ({360/5 * (5 - 1)+\margin}:\radius) 
            arc ({360/5 * (5 - 1)+\margin}:{360/5 * (5)-\margin}:\radius);

        \begin{scope}[shift={(10cm, 0cm)}]
    
            \def \n {5}
        \def \radius {4cm}
        \def \margin {12} 
        
        \node[draw, white, circle, fill=SeaGreen!70, minimum size=0.8cm] at ({360/5 * (1 - 1)}:\radius) {51};
        \draw[-, >=latex] ({360/5 * (1 - 1)+\margin}:\radius) 
            arc ({360/5 * (1 - 1)+\margin}:{360/5 * (1)-\margin}:\radius);
        
        \node[draw, white, rounded corners, rectangle, fill=SeaGreen, minimum size=0.8cm] at ({360/5 * (2 - 1)}:\radius) {\Large{50}};
        \draw[-, >=latex] ({360/5 * (2 - 1)+\margin}:\radius) 
            arc ({360/5 * (2 - 1)+\margin}:{360/5 * (2)-\margin}:\radius);
        
        \node[draw, black, circle, fill=white, minimum size=0.8cm] at ({360/5 * (3 - 1)}:\radius) {};
        \draw[-, >=latex] ({360/5 * (3 - 1)+\margin}:\radius) 
            arc ({360/5 * (3 - 1)+\margin}:{360/5 * (3)-\margin}:\radius);
        
        \node[draw, black, circle, fill=white, minimum size=0.8cm] at ({360/5 * (4 - 1)}:\radius) {};
        \draw[-, >=latex] ({360/5 * (4 - 1)+\margin}:\radius) 
            arc ({360/5 * (4 - 1)+\margin}:{360/5 * (4)-\margin}:\radius);
        
        \node[draw, white, rounded corners, rectangle, fill=IndianRed, minimum size=0.8cm] at ({360/5 * (5 - 1)}:\radius) {51};
        \draw[-, >=latex] ({360/5 * (5 - 1)+\margin}:\radius) 
            arc ({360/5 * (5 - 1)+\margin}:{360/5 * (5)-\margin}:\radius);
            
        \end{scope}

        \begin{scope}[shift={(20cm, 0cm)}]
    
            \def \n {5}
        \def \radius {4cm}
        \def \margin {12} 

        \node[draw, black, circle, fill=white, minimum size=0.8cm] at ({360/5 * (3 - 1)}:\radius) {};
        \draw[-, >=latex] ({360/5 * (3 - 1)+\margin}:\radius) 
            arc ({360/5 * (3 - 1)+\margin}:{360/5 * (3)-\margin}:\radius);
            
        \draw[-, >=latex] ({360/5 * (1 - 1)+\margin}:\radius) 
            arc ({360/5 * (1 - 1)+\margin}:{360/5 * (1)-\margin}:\radius);
        
        \node[draw, white, rounded corners, rectangle, fill=SeaGreen, minimum size=0.8cm] at ({360/5 * (2 - 1)}:\radius) {\Large{50}};
        \draw[-, >=latex] ({360/5 * (2 - 1)+\margin}:\radius) 
            arc ({360/5 * (2 - 1)+\margin}:{360/5 * (2)-\margin}:\radius);
        
        \node[draw, white, circle, fill=SeaGreen!70, minimum size=0.8cm] at ({360/5 * (1 - 1)}:\radius) {51};
        \draw[-, >=latex] ({360/5 * (3 - 1)+\margin}:\radius) 
            arc ({360/5 * (3 - 1)+\margin}:{360/5 * (3)-\margin}:\radius);
        
        \node[draw, black, circle, fill=white, minimum size=0.8cm] at ({360/5 * (4 - 1)}:\radius) {};
        \draw[-, >=latex] ({360/5 * (4 - 1)+\margin}:\radius) 
            arc ({360/5 * (4 - 1)+\margin}:{360/5 * (4)-\margin}:\radius);
        
        \node[draw, white, rounded corners, rectangle, fill=IndianRed, minimum size=0.8cm] at ({360/5 * (5 - 1)}:\radius) {50};
        \draw[-, >=latex] ({360/5 * (5 - 1)+\margin}:\radius) 
            arc ({360/5 * (5 - 1)+\margin}:{360/5 * (5)-\margin}:\radius);

        \end{scope}
        
    \end{tikzpicture}    
\caption{Some examples illustrating Lata's drawing strategy described above with $101$ troops, where Raj responds to Lata's first move (denoted by a green square) by playing on a non-adjacent location on his first move. \textbf{(Left.)} If Raj places more than $51$ troops (i.e, the number of troops that Lata has after making her first move), Lata responds by keeping distance from Raj's ``heavyweight'' vertex and ensures a win no matter how Raj plays his remaining $49$ troops. \textbf{(Middle.)} If Raj uses $51$ troops instead, Lata responds by matching next to Raj, and now no matter how Raj places his remaining $50$ troops, the game draws. \textbf{(Right.)} If Raj uses $50$ troops or less, Lata responds by matching next to Raj, and depending on how Raj places his remaining $51$ troops, the game either draws or Lata wins.}
\end{figure}

\textbf{A Draw Strategy for Raj.} Suppose Lata begins by placing $x$ troops on $v_1$. If $x = T$, then Raj can draw by placing $T$ troops on, say $v_3$, or win by placing one troop on $v_3$ and $(T-1)$ troops on $v_4$. So we assume that Lata places $x < T$ troops on her first turn.

Raj responds to this move by placing $\lfloor \nicefrac{T}{2} \rfloor$ troops on $v_5$. Lata has $(T-x)$ troops left, which she can either choose to play fully on her second move (and take the first attack), or she can choose to distribute $(T-x) = a + b$ across two moves. We provide responses for both cases separately. First, suppose Lata places $(T-x)$ troops on her turn. 

\begin{itemize}
\item Suppose Lata plays on $v_2$. Then Raj places $\lceil \nicefrac{T}{2} \rceil$ troops on $v_3$.

\item Suppose Lata plays on $v_3$. Then Raj places $\lceil \nicefrac{T}{2} \rceil$ troops on $v_4$.

\item Suppose Lata plays on $v_4$. Then Raj places $\lceil \nicefrac{T}{2} \rceil$ troops on $v_2$.
\end{itemize}

In all cases, Raj draws by an argument similar to the one we made in the proof of~\Cref{lem:2edges-draw}. Now, suppose Lata places $\lceil \nicefrac{T}{2} \rceil \leqslant y < (T-x)$ troops on her second move. We then have the following cases.

\begin{itemize}
\item Suppose Lata plays on $v_2$. Then Raj places $\lceil \nicefrac{T}{2} \rceil$ troops on $v_4$. It is easy to see that Raj draws in this situation.

\item Suppose Lata plays on $v_3$. Then Raj places $\lceil \nicefrac{T}{2} \rceil$ troops on $v_2$. This forces Lata to play her remaining troops on $v_4$. Note that she plays fewer than $\lfloor \nicefrac{T}{2} \rfloor$ troops on $v_4$, since she has already used up at least $\lceil \nicefrac{T}{2} \rceil$ troops on $v_4$, and at least one on $x$. Raj attacks $v_1$ on his first turn in the attack phase, and $v_4$ on his second turn, and thus the game ends in a draw.

\item Suppose Lata plays on $v_4$. Then Raj places $\lceil \nicefrac{T}{2} \rceil$ troops on $v_2$. The argument for why this is a draw is similar to the previous case.
\end{itemize}

If Lata places $y < \lceil \nicefrac{T}{2} \rceil$ troops on her second turn, we have the following cases.

\begin{itemize}
\item Suppose Lata plays on $v_2$. Then Raj places $\lceil \nicefrac{T}{2} \rceil$ troops on $v_3$, and Lata places her remaining troops on $v_4$. Notice that both $v_4$ and $v_2$ are vulnerable in this situation, and Raj can pull off two valid attacks while Lata has at most one, and the game is drawn.

\item Suppose Lata plays on $v_3$. Then Raj places $\lceil \nicefrac{T}{2} \rceil$ troops on $v_2$, and Lata places her remaining troops on $v_4$. Notice that both $v_1$ and $v_3$ are vulnerable in this situation, and Raj can pull off two valid attacks while Lata has at most one, and the game is drawn.

\item Suppose Lata plays on $v_4$. Then Raj places $\lceil \nicefrac{T}{2} \rceil$ troops on $v_2$, and Lata places her remaining troops on $v_3$. Now, on his first move in the attack phase, Raj attacks $v_1$, and after this, the argument for drawing the remaining game is the same as the argument used in~\Cref{lem:2edges-draw}.
\end{itemize}
This concludes the description of the draw strategy for Raj.
\end{proof}

\section{Concluding Remarks}
\label{sec:conclusions}

We initiated an analysis for the game of~\textsc{Aggression}, and established that the problem of devising a winning response to a given troop placement \emph{and} attack plan is \textsf{NP}-complete. We also showed that both players can draw on matchings if there are at least as many edges as troops or if there are at most two edges to begin with. We also showed that the second player can win on matchings that have at least three edges if the number of troops is at least six more than the number of edges. Finally, we also showed that both players can draw the game if the graph is a cycle of length at most five. There are several considerations for further work, and we suggest some open problems below.

\begin{itemize}
    \item What is the complexity of~\textsc{Optimal Response} when the graph is planar or acyclic?
    \item What is the complexity of~\textsc{Optimal Response} for the single-troop variant of the game?
    \item What is the complexity of the problem of finding which player wins a game if a graph and a $(T_L, T_R)$-troop placement is given?
    \item Which player has a winning strategy on a matching with $N$ edges and $T$ troops where $T \in \{N+1, N+2, N+3, N+4, N+5\}$?
    \item Which player wins the game on paths and cycles of length $k$?
\end{itemize}

We also propose that the following variant of the game is interesting: when both players can place at most one troop on their turn. This version, which we call \textsc{Micro Aggression}, appears to be non-trivial even on paths and cycles. We can establish, using standard mirroring arguments, that the second player can draw on odd cycles of length at least five and that the first player can draw on a path of any length, as shown below:

\begin{lemma} The first player can draw \textsc{Micro Aggression} on any path.
\end{lemma}

\begin{proof}
   The first player begins his turn by playing at the centre in the case of odd paths and an arbitrary vertex in the case of even paths. Imagine a mirror placed in the center of the path (through the center vertex in case of odd paths and after the $\frac{n}{2}^{th}$ vertex in the case of even paths). Now whenever the second player makes a move, first player places his troop on the reflection of the vertex where the second player just plays his troops unless the second player has played on the reflection of the first player's last move. The first player places his troop on an arbitrary vertex in this case.
Since it can be seen that first player never plays on the reflection of one of his previously played vertex, the configuration looks such that the center vertex has a first player's troop and any other vertex has a first player's troop if and only if its reflection has a second player's troop. Now in case of even paths, first player starts with an attack across the center mirror if such an attack exists and an arbitrary attack otherwise. He mimics the second player's attack on the opposite side in each subsequent turn. If the second player mirrors the first player's attack, first player can attack an arbitrary vertex. In the case of odd paths, the second player attacks first and the first player can always mirror him from the opposite side. Note that the center vertex belongs to the first player and each one of its neighbours belong to the first player and the second player. So any attack must happen completely on one side of the mirror. In both the places, the first player never runs out of attacks before the first player so at the end, the second player cannot have more territories than the first player.
\end{proof}

\begin{lemma} The second player can draw \textsc{Micro Aggression} on any odd-length cycle with at least five vertices.
\end{lemma}

\begin{proof}
    The second player fixes a vertex as the centre of the cycle and imagines a mirror across it. He mirrors first player's every move unless the first player playes at the centre or at a reflection of one of his previously played vertices and he plays at an arbitrary vertex in this case. After all the troops are placed, the second player attacks first so if there is a valid attack across the center the second player plays that attack and starts with an arbitrary attack otherwise. He mimics the first player's attack on the opposite side in each subsequent turn. If the first player mirrors the second player's attack, second player can attack an arbitrary vertex. The second player never runs out of attacks before the first player and thus he cannot lose.
\end{proof}

We can also show, by an explicit strategy, that the first player wins on a path of length five. We leave all other cases open.





\end{document}